\documentclass[11pt,a4paper]{article}

\usepackage[a4paper,left=2cm,right=2cm,top=3cm,bottom=3cm]{geometry}
\usepackage[ansinew]{inputenc}  
\usepackage[german,english]{babel}
\usepackage{latexsym,amsbsy,amssymb,amsmath,amsfonts}
\usepackage{epsfig}
\usepackage{euscript}
\usepackage{units}
\usepackage{pdflscape}
\usepackage{indent}
\usepackage{enumitem}
\usepackage{rotating}

\usepackage{xcolor}

\usepackage{tikz}
\usetikzlibrary{shapes}
\usetikzlibrary{automata}
\usetikzlibrary{arrows}
\usetikzlibrary{calc}
\usepackage[colorlinks,linkcolor={red!75!black},citecolor={blue!75!black},urlcolor={blue!80!black}]{hyperref}

\newcommand{\opstyle}[1]{\mathrm{#1}}

\newcommand{\redstyle}[1]{\mathnormal{#1}}

\newcommand{\pairing}[1]{\langle #1 \rangle}
\newcommand{\oli}{\overline}

\newcommand{\cP}{\opstyle{P}}
\newcommand{\cNP}{\opstyle{NP}}
\newcommand{\ccoNP}{\opstyle{coNP}}
\newcommand{\cUP}{\opstyle{UP}}

\newcommand{\FP}{\opstyle{FP}}

\newcommand{\calT}{{\cal T}}

\newcommand{\reduction}[3][]{%
    \redstyle{\le_{\mathrm{#3}}^{\mathrm{#2}#1}}}
\newcommand{\polyreduction}[2][]{\reduction[#1]{p}{#2}}
\newcommand{\redm}[1][]{\polyreduction[#1]{m}}

\newtheorem{dummytheorem}{Dummy-Theorem}[section]
\newcommand{\proofendsign}{$\Box$} 
\newtheorem{definition}[dummytheorem]{Definition}
\newtheorem{lemma}[dummytheorem]{Lemma}
\newtheorem{theorem}[dummytheorem]{Theorem}

\newtheorem{proposition}[dummytheorem]{Proposition}
\newtheorem{corollary}[dummytheorem]{Corollary}

\newtheorem{claim}[dummytheorem]{Claim}

\newenvironment{proof}{{\noindent \bf Proof }}{%
    {\hspace*{\fill}\proofendsign\par\bigskip}}

\newcommand{\isdefinedl}{\mathop{=}\limits^{\mbox{%
    \raisebox{-0.15ex}[0ex][0ex]{$\scriptscriptstyle df$}}}}

\newcommand{\isdefined}{\isdefinedl}

\newcommand{\az}{{\Sigma}} 
\newcommand{\sow}{{\az^*}}

\newcommand{\sowle}[2][\az]{{{#1}^{\le #2}}}

\newcommand{\rem}[1]{}

\newcommand{\tn}[1]{\textnormal{#1}}

\newlength{\chrlengtha}
\newlength{\chrlengthb}

\newcommand{\N}{\mathbb{N}}
\newcommand{\Z}{\mathbb{Z}}

\begin{document}
\selectlanguage{english}

\def\sqsubsetneq{\mathrel{\sqsubseteq\kern-0.92em\raise-0.15em\hbox{\rotatebox{313}{\scalebox{1.1}[0.75]{\(\shortmid\)}}}\scalebox{0.3}[1]{\ }}}
\def\sqsupsetneq{\mathrel{\sqsupseteq\kern-0.92em\raise-0.15em\hbox{\rotatebox{313}{\scalebox{1.1}[0.75]{\(\shortmid\)}}}\scalebox{0.3}[1]{\ }}}

\newcommand{\DisjcoNP}{\tn{DisjCoNP}}

\title{P$\ne$NP and All Non-Empty Sets in $\cNP\cup \ccoNP$\\ Have P-Optimal Proof Systems Relative to an Oracle}

\author{Titus Dose\\Julius-Maximilians-Universität Würzburg}
\maketitle

\begin{abstract}
    \noindent
    As one step in a working program initiated by Pudlák \cite{pud17}
    we construct an oracle relative to which $\cP\ne\cNP$ and
    all non-empty sets in $\cNP\cup\ccoNP$ have $\cP$-optimal proof systems.    
\end{abstract}

\section{Introduction}
The main motivation for the present paper is an article by Pudlák \cite{pud17} 
who lists several major conjectures in the field
of proof complexity and discusses their relations. 
Among others, Pudl{\'a}k conjectures the following assertions (note that within the present
paper all reductions are polynomial-time-bounded):
\begin{itemize}
	\item   $\mathsf{CON}$ (resp., $\mathsf{SAT}$): $\ccoNP$ (resp., $\cNP$) does not contain many-one complete
            sets that have P-optimal proof systems 
    \item   $\mathsf{CON}^{\mathsf{N}}$: $\ccoNP$ does not contain many-one complete
            sets that have optimal proof systems,\\
            (note that $\mathsf{CON}^{\mathsf{N}}$ is the non-uniform version of
            $\mathsf{CON}$)
    \item   $\mathsf{DisjNP}$ (resp., $\mathsf{DisjCoNP}$): The class of all disjoint $\cNP$-pairs 
            (resp., $\ccoNP$-pairs) does not have many-one complete elements,
    \item   $\mathsf{TFNP}$: The class of all total polynomial search problems does not have complete elements,
    \item   $\mathsf{NP}\cap\mathsf{coNP}$ (resp., $\mathsf{UP}$): $\cNP\cap\ccoNP$ (resp., $\cUP$, 
            the class of problems accepted by $\cNP$ machines with at most one accepting path
            for each input) does not have many-one complete elements.
\end{itemize}
Pudlák asks for oracles separating corresponding relativized conjectures.
Recently there has been made some progress in this working program
\cite{kha19,dg19,dos19} which is documented by the following figure 
representing the current state of the art.
\newcommand{\CCON}{\mathsf{CON}}
\newcommand{\CCONN}{\mathsf{CON}^\mathsf{N}}
\newcommand{\CDisjNP}{\mathsf{DisjNP}}
\newcommand{\CUP}{\mathsf{UP}}
\newcommand{\CRFN}{\mathsf{RFN}_1}
\newcommand{\CPNP}{\mathsf{P}\ne\mathsf{NP}}
\newcommand{\CNPcoNP}{\mathsf{NP}\cap\mathsf{coNP}}
\newcommand{\CSAT}{\mathsf{SAT}}
\newcommand{\CTFNP}{\mathsf{TFNP}}
\newcommand{\CDisjCoNP}{\mathsf{DisjCoNP}}

\begin{center}
\begin{figure}[ht]
\begin{tikzpicture}[->,>=stealth',initial text={},shorten >=1pt,auto,node distance=2cm,
  thin,main node/.style={draw=none,font=\sffamily\Large}, thin]

  \node[main node] (1)  {$\CDisjNP$};
  \node[main node] (2) [below of=1] {$\CCONN$};
  \node[main node] (3) [right  of=2] {$\CUP$};
  \node[main node] (4) [ below of=3] {$\CCON$};
\node[main node] (12) [right of=4]{};
\node[ main node] (14) [ below of=4] {};
\node[ main node] (15) [ right of=14] {};
	\node[ main node] (5) [ right of=15] {$\CCON\vee\CSAT$};
  \node[ main node] (6) [ below  of=5] {$\CPNP$};
  \node[ main node] (7) [ right of=12] {$\CNPcoNP$};
\node[main node] (13) [right of=7]{};
	\node[ main node] (8) [ right of=13] {$\CSAT$}; 
	\node[ main node] (9) [ above of=8] {$\CTFNP$};
	\node[ main node] (10) [ above of=9] {$\CDisjCoNP$};

  \path[every node/.style={font=\sffamily\small, color  = black}]
    (1) edge[bend left] node[left] {} (2)
     (2)   edge [] node [] {} (4)
  (3) edge node [] {} (4)
  (4) edge node [] {} (5)
	(5) edge node [] {} (6)
	(7) edge node [] {} (5)
	(8) edge node [] {} (5)
    (9) edge node [] {} (8)
    (10) edge [] node  {} (9);
    
    \path[every node/.style={font=\sffamily\small, color  = black},dashed]
    (1) edge node [above,sloped] {\cite{dg19}} (3)
    (7) edge node [above] {\cite{dos19b}} (4)
    (1) edge [bend left] node [above,sloped] {\cite{dos19}} (8)
    (3) edge [bend left] node [sloped,above] {\cite{dos19}\phantom{........}} (8)
    (7) edge [] node [above] {\cite{dos19}} (8)
    (10) edge [] node [above,sloped] {\cite{kha19}\phantom{...............}} (4)
    (2) edge[bend left] node[left]{\cite{gssz04}}(1)
    (6) edge[bend right=60] node[right] {Cor \ref{cor_14792314}}(5);
\end{tikzpicture}
\caption{\label{fig_1047120471}
Solid arrows mean implications. All implications occurring in
the graphic have relativizable proofs. A dashed arrow from
one conjecture $\mathsf{A}$ to another conjecture $\mathsf{B}$ means that
there is an oracle $X$ against the implication $\mathsf{A}\Rightarrow
\mathsf{B}$, i.e., relative to $X$, it holds $\mathsf{A}\wedge\neg
\mathsf{B}$.
\newline Pudl{\'a}k \cite{pud17} also defines the conjecture $\CRFN$ and lists it between $\CCON\vee\CSAT$ and
$\CPNP$, i.e., $\CCON\vee\CSAT\Rightarrow \CRFN\Rightarrow \CPNP$. Khaniki \cite{kha19} even shows $\CCON\vee\CSAT\Leftrightarrow \CRFN$, which is why we omit $\CRFN$ in the figure. 
For a definition of $\CRFN$ we refer to \cite{pud17}.
}
\end{figure}
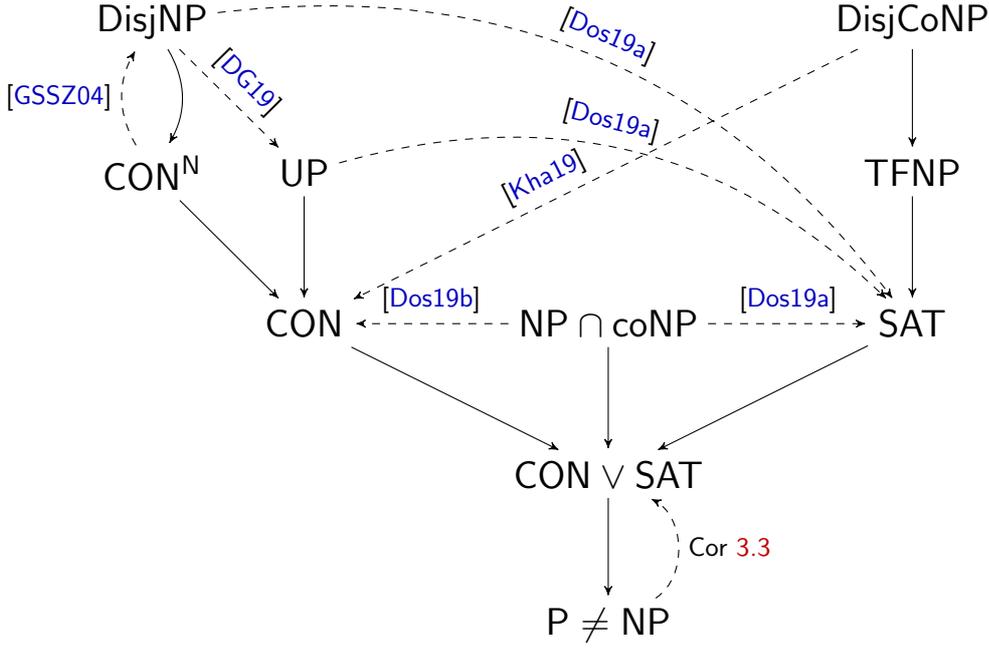
\end{center}

In the figure $O$ denotes the oracle that we construct in the present paper.
It shows that there is no relativizable proof for the implication
$\CPNP \Rightarrow \CCON\vee\CSAT$, i.e.\ the conjectures $\CPNP$ and
$\CCON\vee\CSAT$ cannot be shown equivalent with relativizable proofs.
More precisely, the relativization of $\CCON\vee\CSAT$ (i.e., the statement ``for all oracles
$D$ it holds (i) there is no $A\in\cNP^D$ that has $\cP^D$-optimal proof systems 
or (ii) there is no $A\in\ccoNP^D$ that has $\cP^D$-optimal proof systems'')
is strictly stronger than the relativization of $\CPNP$ (i.e., the statement ``for
all oracles $D$ it holds $\cP^D\ne\cNP^D$'').

\section{Preliminaries} \label{sec_prelim}
Most parts of this section are copied from our previous papers
\cite{dg19} and \cite{dos19}.

Throughout this paper let $\Sigma$ be the alphabet $\{0,1\}$.
We denote the length of a word
$w\in\sow$ by $|w|$.
Let $\sowle{n} = \{w \in \sow ~|~ |w| \le n\}$.
The empty word is denoted by $\varepsilon$ and
the $i$-th letter of a word $w$ for $0 \le i < |w|$ is denoted by $w(i)$, i.e.,
$w = w(0) w(1) \cdots w(|w|-1)$.
If $v$ is a prefix of $w$, i.e., $|v|\le|w|$ and $v(i) = w(i)$ 
for all $0\le i <|v|$, then we write $v \sqsubseteq w$.
For each finite set $Y \subseteq \sow$, let
$\ell(Y) \isdefined \sum_{w \in Y} |w|$.

The set of all integers is denoted by $\Z$.
Moreover, $\N$ denotes the set of natural numbers and
$\N^+$ denotes the set of positive natural numbers.
The identity function $x\mapsto x$ is denoted by
$id$.

We identify $\sow$ with $\N$
via the polynomial-time computable, polynomial-time invertible bijection
$w \mapsto \sum_{i<|w|} (1+w(i)) 2^{|w|-1-i}$,
which is a variant of the dyadic encoding.
Hence, notations, relations, and operations for $\sow$
are transferred to $\N$ and vice versa.
In particular, $|n|$ denotes the length of $n \in \N$.
We eliminate the ambiguity of the expressions $0^i$ and $1^i$
by always interpreting them over $\sow$.

Let $\langle \cdot \rangle : \bigcup_{i \ge 0} \N^i \rightarrow \N$
be an injective, polynomial-time computable, polynomial-time invertible
pairing function such that
$|\pairing{u_1, \ldots, u_n}| = 2(|u_1| + \cdots + |u_n| + n)$.

Given two sets 
$A$ and $B$, $A-B$ denotes the set difference
between $A$ and $B$. 
The complement of a set $A$ relative to the universe $U$ 
is denoted by $\oli{A}= U-A$. The universe will always be
apparent from the context. The symmetric difference of sets $A$ and $B$
is denote by $A\triangle B = (A-B)\cup (B-A)$.

$\FP$, $\cP$, and $\cNP$ denote standard complexity classes \cite{pap94}.
Define $\opstyle{co\mathcal{C}} = \{A\subseteq\Sigma^*\mid \oli{A} \in\mathcal{C}\}$
for a class $\mathcal{C}$.
We also consider all these complexity classes in the presence of
an oracle $O$ and denote the corresponding classes by $\FP^O$, $\cP^O$, $\cNP^O$, and so on.

Let $M$ be a Turing machine. $M^D(x)$ denotes the 
computation of $M$ on input $x$ with $D$ as an oracle. For an
arbitrary oracle $D$ we let $L(M^D) = \{ x ~|~ M^D(x) \tn{ accepts}\}$.
A nondeterministic computation accepts if and only if it has an
accepting path.

For a deterministic polynomial-time Turing transducer, depending on the context,
$F^D(x)$ either denotes the computation of $F$ on input $x$ with $D$ as an oracle
or the output of this computation.

\begin{definition}\label{definition_019470123712}
A sequence $(M_i)$ is called {\em standard enumeration} of
nondeterministic, polynomial-time oracle Turing machines,
if it has the following properties:
\begin{enumerate}
    \item All $M_i$ are nondeterministic, polynomial-time oracle Turing machines.
    \item For all oracles $D$ and all inputs $x$
    the computation $M_i^D(x)$ stops within $|x|^i + i$ steps.
    \item For every nondeterministic, polynomial-time oracle Turing machine $M$
    there exist infinitely many $i \in \N$ such that
    for all oracles $D$ it holds that $L(M^D) = L(M_i^D)$.
    \item There exists a nondeterministic, polynomial-time oracle Turing machine $M$
    such that for all oracles $D$ and all inputs $x$ it holds that $M^D(\langle i,0^{|x|^i+i},x \rangle)$
    nondeterministically simulates the computation $M_i^D(x)$.
\end{enumerate}
Analogously we define standard enumerations of
deterministic, polynomial-time oracle Turing transducers.
\end{definition}
Throughout this paper, we fix some standard enumerations.
Let $M_1,M_2,\dots$ be a standard enumeration of nondeterministic
polynomial-time oracle Turing machines. Then for every oracle $D$, the sequence
$(M_i)$ represents an enumeration of languages in $\cNP^D$, i.e.,
$\cNP^D = \{L(M_i^D)\mid i\in\N^+\}$.
Let $F_1,F_2,\dots$ be a standard enumeration of polynomial time 
oracle Turing transducers.

By the properties of standard enumerations, for each oracle $D$ the
problem
\begin{equation*}K^D = \{\pairing{0^i,0^t,x}\mid \tn{$M_i^D(x)$ accepts within $t$ steps}  \}\label{eq_14702}
\end{equation*}
is $\cNP^D$-complete (in particular it is in $\cNP^D$).

\newcommand{\psim}[1][]{\le^{\mathrm{p}#1}}
\begin{definition}[\cite{cr79}]
    A function $f \in \FP$ is called {\em proof system} for the set $\tn{ran}(f)$.
    For $f,g \in \FP$ we say that {\em $f$ is simulated by $g$} (resp., 
    {\em $f$ is $\cP$-simulated by $g$}) denoted by $f\le g$ (resp., $f \psim g$),
    if there exists a function $\pi$ (resp., a function $\pi \in \FP$)
    and a polynomial $p$ such that $|\pi(x)|\le p(|x|)$ and $g(\pi(x)) = f(x)$ for all $x$.
    A function $g \in \FP$ is {\em optimal} (resp., {\em $\cP$-optimal}),
    if $f\le g$ (resp., $f \psim g$) for all $f \in \FP$ with $\tn{ran(f)}=\tn{ran(g)}$.
    Corresponding relativized notions are obtained by using
    $\cP^O$, $\FP^O$, and $\psim[,O]$ in the definitions above.
\end{definition}The following proposition states the relativized version of a result by
K{\"o}bler, Messner, and Tor\'an \cite{kmt03},
which they show with a relativizable proof.
\begin{proposition}[\cite{kmt03}] \label{propo_pps_oracle}
    For every oracle $O$,
    if $A$ has a $\cP^O$-optimal (resp., optimal) proof system
    and $B \redm[,O]\! A$,
    then $B$ has a $\cP^O$-optimal (resp., optimal) proof system.
\end{proposition}

\begin{corollary} \label{coro_pps_oracle}
    For every oracle $O$,
    \begin{enumerate}
        \item if there exists a  $\redm[,O]\!$-complete $A \in \cNP^O$
        that has optimal (resp., $\cP^O$-optimal) proof systems,
        then all sets in $\cNP^O$ have optimal (resp., $\cP^O$-optimal) proof systems.
        \item if there exists a  $\redm[,O]\!$-complete $A \in \ccoNP^O$ 
        that has optimal (resp., $\cP^O$-optimal) proof systems,
        then all sets in $\ccoNP^O$ have optimal (resp., $\cP^O$-optimal) proof systems.
    \end{enumerate}
    
\end{corollary}

\medskip
We introduce some quite specific notations that are designed
for the construction of oracles.
The domain and range of a function $t$ are denoted by
$\tn{dom}(t)$ and $\tn{ran}(t)$, respectively.
If a partial function $t$ is not defined at point $x$,
then $t \cup \{x \mapsto y\}$ denotes the continuation $t'$ 
of $t$ that at $x$ has value $y$ and satisfies 
$\tn{dom}(t') = \tn{dom}(t)\cup\{x\}$

If $A$ is a set, then $A(x)$ denotes the characteristic function at point $x$,
i.e., $A(x)$ is $1$ if $x \in A$, and $0$ otherwise.
An oracle $D \subseteq \N$ is identified with its characteristic sequence
$D(0) D(1) \cdots$, which is an $\omega$-word.
(In this way, $D(i)$ denotes both, the characteristic function at point $i$ and
the $i$-th letter of the characteristic sequence, which are the same.)
A finite word $w$ describes an oracle that is partially defined, i.e.,
only defined for natural numbers $x<|w|$.
We can use $w$ instead of the set $\{i ~|~ w(i)=1 \}$ and
write for example $A = w \cup B$, where $A$ and $B$ are sets.
For nondeterministic oracle Turing machines $M$
we use the following phrases:
A computation $M^w(x)$ {\em definitely accepts},
if all paths accept and all queries are $<|w|$.
A computation $M^w(x)$ {\em definitely rejects},
if it contains a path that rejects (within $t$ steps) and
the queries on this path are $<|w|$.
For deterministic oracle Turing machines $P$ we say:
A computation $P^w(x)$ {\em definitely accepts}
(resp., {\em definitely rejects}),
if it accepts (resp., rejects) and
the queries are $<|w|$.

For a deterministic or nondeterministic Turing machine $M$ we say 
that the computation $M^w(x)$ {\em is defined},
if it definitely accepts or definitely rejects.
For a transducer $F$, the computation $F^w(x)$ {\em is defined},
if all queries are $<|w|$.

\section{Oracle Construction}
We now construct the announced oracle.
\begin{lemma}[\cite{dg19}]\label{84189612339821692813}
For all $y\le |w|$ and all $v\sqsupseteq w$ it holds $K^v(y) = K^w(y)$.
\end{lemma}
\begin{proof} We may assume $y = \pairing{0^i,0^t,x}$ for suitable $i,t,x$,
        since otherwise $K^w(y) = K^v(y) = 0$.
        For each $q$ that is queried within the first $t$ steps
        of $M_i^w(x)$ or $M_i^v(x)$ it holds that
        $|q| \le t < |y|$ and thus, $q < y$.
        Hence, these queries are answered the same way relative to $w$ and $v$,
        showing that $M_i^w(x)$ accepts within $t$ steps if and only if $M_i^v(x)$ accepts
        within $t$ steps.
\end{proof}

\begin{theorem}\label{theorem_0917240914}
There exists an oracle $O$ relative to which 
the following statements hold:
\begin{itemize}
  \item $\cP^O\ne\cNP^O$
  \item $K^O$ has $\cP^O$-optimal proof systems.
  \item $\oli{K^O}$ has $\cP^O$-optimal proof systems.
\end{itemize}
\end{theorem}
The following corollary follows from Theorem~\ref{theorem_0917240914} and
Corollary~\ref{coro_pps_oracle}.
\begin{corollary}\label{cor_14792314}
There exists an oracle $O$ relative to which 
the following statements hold:
\begin{itemize}
  \item $\cP^O\ne\cNP^O$
  \item Each set in $\cNP^O$ has $\cP^O$-optimal proof systems.
  \item Each set in $\ccoNP^O$ has $\cP^O$-optimal proof systems.
\end{itemize}
\end{corollary}
\begin{proof}{\bf of Theorem~\ref{theorem_0917240914}}
We define $c(i,x,y) = \pairing{0^i, 0^{|x|^i + i}, x,y}$.
Let $D$ be a (possibly partial) oracle and define
\begin{eqnarray*}
A^D &=& \{0^n\mid \exists_{y\in\az^n} 0y\in D\}.
\end{eqnarray*}
We will construct the oracle such that $A^O\in\cNP^O-\cP^O$ for
the final oracle $O$. Note that throughout this proof we sometimes omit the oracles in
the superscript, e.g., we write $\cNP$ or
$A$ instead of $\cNP^D$ or $A^D$. However, 
we do not do that in the ``actual'' proof but only when 
explaining ideas in a loose way
in order to give the reader the intuition behind the 
occasionally very technical arguments.

Let us briefly sketch the idea of our construction.

{\em Preview of construction.} For each $F_i$ we first try to ensure that
$F_i$ does not compute a proof system for $K$ (resp., $\oli{K}$). If this is
impossible, then $F_i$ inherently computes a proof system for $K$ (resp., $\oli{K}$).
In that case we start to encode the values of $F_i$ into the oracle so that $F_i$ can
be $\cP$-simulated by some proof system for $K$ (resp., $\oli{K}$)
that we will define later and finally show to be $\cP$-optimal.

Moreover, we diagonalize against all $P_i$ such that $A$ is not in $\cP$ relative to the final oracle.

\begin{claim}[\cite{dg19}] \label{claim_041730745333}
        Let $w\in\az^*$ be an oracle, $i\in\N^+$, and $x,y\in\N$ such that
        $c(i,x,y)\le |w|$. Then the following holds.
        \begin{enumerate}
            \item $F_i^w(x)$ is defined and $F_i^w(x)<|w|$.
            \item For all $v \sqsupseteq w$, $\big(F_i^w(x)\in K^w\Leftrightarrow F_i^w(x)\in K^v\big)$.
            
        \end{enumerate}
    \end{claim}
\begin{proof}
As the running time of $F_i^w(x)$ is bounded by 
$|x|^i + i < |c(i,x,y)| < c(i,x,y) \le |w|$,
the computation $F_i^w(x)$ is defined and its output 
is less than $|w|$. Hence, 1 holds.
Consider 2. It suffices to show that 
$K^v(q) = K^w(q)$ for all $q < |w|$ and all $v\sqsupseteq
w$. This holds by Lemma~\ref{84189612339821692813}.
\end{proof}

During the construction we maintain a collection 
of requirements $t: \{0,1\}\times\N^+\to \N$, 
where $t\in\calT$ for $$\calT = \{t:\{0,1\}\times\N^+\to \N\mid
\text{$t$ has a finite domain}\}.$$
A partial oracle $w$ is called $t$-valid 
if it satisfies the following properties.
\begin{itemize}
    \item[V1]   For all $i\in\N^+$,
                \begin{enumerate}
                    \item  if $10c(i,x,y)\in w$ for some $x,y\in\N$, then $F_i^w(x) = y\in K^w$.
                    \item  if $11c(i,x,y)\in w$ for some $x,y\in\N$, then $F_i^w(x) = y\in \oli{K^w}$.
                \end{enumerate}
                
    \item[V2]   For all $i\in\N^+$, if $t(0,i) = 0$, then there exists $x$ such that $F_i^w(x)$ is defined  
                and $F_i^w(x)\notin K^v$ for all $v\sqsupseteq w$.
                              
    \item[V3]   For all $i\in\N^+$, if $t(0,i)>0$, then for all $x\in\N$ with 
                $t(0,i) \le 10c(i,x,F_i^w(x)) < |w|$,
                it holds $10c(i,x,F_i^w(x))\in w$.
                          
    \item[V4]   For all $i\in\N^+$, if $t(1,i) = 0$, then there exists $x$ such that $F_i^w(x)$ is defined
                and $F_i^w(x)\notin \oli{K^v}$ for all $v\sqsupseteq w$.
                
    \item[V5]   For all $i\in\N^+$, if $t(1,i)>0$, then for all $x\in\N$ with 
                $t(1,i) \le 11c(i,x,F_i^w(x)) < |w|$,
                it holds $11c(i,x,F_i^w(x))\in w$.       
\end{itemize}
The following claim follows directly from the definition of $t$-valid.
\begin{claim}\label{claim_10794190247210}
Let $t,t'\in\calT$ such that $t'$ is an extension of $t$. If $w\in\az^*$
is $t'$-valid, then $w$ is $t$-valid.
\end{claim}
\begin{claim}\label{claim_1094101243}
Let $t\in\calT$ $u,v,w\in\az^*$ be oracles with $u\sqsubseteq v\sqsubseteq w$. 
If $u$ and $w$ are $t$-valid, then $v$ is $t$-valid.
\end{claim}
\begin{proof}
$v$ satisfies V2 and V4 since $u$ satisfies V2 and V4.

Let us argue for V1. Let $10c(i,x,y)\in v$ for $i\in\N^+$ and $x,y\in\N$.
Then $10c(i,x,y)\in w$ and as $w$ is $t$-valid, it holds by V1 that $F_i^w(x) = y\in K^w$.
By Claim~\ref{claim_041730745333}, $F_i^v(x) = F_i^w(x) = y$ and $K^v(y) = K^w(y) = 1$.
Analogously, $11c(i,x,y)\in v$ for $i\in\N^+$ and $x,y\in\N$ implies $F_i^v(x) = y\in \oli{K^v}$.
Thus, $v$ satisfies V1.

Consider V3/V5. Let $i\in\N^+$, $x\in\N$, and $b\in \{0,1\}$ such that $0<t(b,i)\le 1bc(i,x,F_i^v(x)) < |v|$.
Then by Claim~\ref{claim_041730745333}, $F_i^v(x) = F_i^w(x)$.
As $w$ is $t$-valid, we obtain by V3/V5 that $1bc(i,x,F_i^w(x))\in w$. Since $1bc(i,x,F_i^v(x)) < |v|$
and $v\sqsubseteq w$, we have $1bc(i,x,F_i^v(x)) = 1bc(i,x,F_i^w(x))\in v$, which shows that
$v$ satisfies V3/V5. 
\end{proof}

{\em Oracle construction.} Let $T:\N\to \{0,1,2\}\times\N^+$ be a bijection.
Each value of $T(s)$ for $s\in\N$ stands for a task. 
We treat the tasks in the order specified by $T$.
We start with the nowhere defined function $t_0$ and the $t_0$-valid oracle $w_0 = 
\varepsilon$. Then we define functions $t_1,t_2,\dots$ in $\calT$ such that $t_{i+1}$
is an extension of $t_i$ and partial oracles $w_0\sqsubsetneq w_1\sqsubsetneq w_2
\sqsubsetneq\dots$ such that each $w_i$ is $t_i$-valid. Finally, we choose
$O = \bigcup_{i = 0}^\infty w_i$ (note that $O$ is totally defined since in
each step we strictly extend the oracle).
We describe step $s>0$, which starts with a $t_{s-1}$-valid oracle $w_{s-1}$
and extends it to a $t_s$-valid $w_s\sqsupsetneq w_{s-1}$ depending on the value
of $T(s)$. We will argue later that the construction is possible.
\begin{itemize}
    \item task $(0,i)$ for $i\in\N^+$: Let $t' = 
        t_{s-1}\cup \{(0,i)\mapsto 0\}$. If there
        exists a $t'$-valid $v\sqsupsetneq w_{s-1}$, 
        then let $t_s = t'$ and define $w_s = v$ for the least 
        $t'$-valid $v\sqsupsetneq w_{s-1}$.
        Otherwise, let
        $t_s = t_{s-1}\cup \{(0,i)\mapsto |w_{s-1}|\}$
        and choose $w_s = w_{s-1}b$ for $b\in\{0,1\}$
        such that $w_s$ is $t_s$-valid. 
    \item task $(1,i)$ for $i\in\N^+$: 
        Let $t' = t_{s-1}\cup \{(1,i)\mapsto 0\}$. If there
        exists a $t'$-valid $v\sqsupsetneq w_{s-1}$, 
        then let $t_s = t'$ and define $w_s = v$ for the least 
        $t'$-valid $v\sqsupsetneq w_{s-1}$.
        Otherwise, let  
        $t_s = t_{s-1}\cup \{(1,i)\mapsto |w_{s-1}|\}$
        and choose $w_s = w_{s-1}b$ for $b\in\{0,1\}$
        such that $w_s$ is $t_s$-valid. 
    \item task $(2,i)$ for $i\in\N^+$: Let $t_s = t_{s-1}$ and
        chose $w_s\sqsupsetneq w_{s-1}$ such that for some $n\in\N$ the
        computation $P_i^{w_s}(0^n)$ is defined, $0^n\in A^v\Leftrightarrow
        0^n\in A^{w_s}$ for all $v\sqsupseteq w_s$, and $\big(0^n\in A^{w_s} \Leftrightarrow 
        P_i^{w_s}(0^n)\text{ rejects}\big)$.
        
\end{itemize}

\begin{claim}\label{claim_82146969998}
Let $s\ge 0$ and $w\sqsupseteq w_s$ such that $w$ is $t_s$-valid.
    \begin{enumerate}
        \item\label{item_01283470} If $z = 10c(i,x,F_i^w(x))$ for $i\in\N^+$ and $x\in\N$ with 
            $0 < t_s(0,i) \le z$, then $w1$ is $t_s$-valid.
        \item\label{item_10278480} If $z = 11c(i,x,F_i^w(x))$ for $i\in\N^+$ and $x\in\N$ with 
            $0 < t_s(1,i) \le z$, then $w1$ is $t_s$-valid.
        \item\label{item_389217401} If $z = 0y$ for $y\in\az^n$ and $n\in\N$, then $w0$ and $w1$ are $t_s$-valid.
        \item\label{item_4712800032} In all other cases (i.e., none of the assumptions 
            in \ref{item_01283470}--\ref{item_389217401}
            holds) $w0$ is $t_s$-valid.
    \end{enumerate}
\end{claim}
\begin{proof}
First observe that V2 and V4 are not affected by extending the oracle.
Moreover, by Claim~\ref{claim_041730745333}, as $w$ satisfies V1, V3, and V5,
$wb$ for $b\in\{0,1\}$ satisfies
\begin{itemize}
	\item[(A)]  V1.1 unless $b = 1$, $z = 10c(i,x,y)$ for $i,x,y\in\N$ with $i>0$ and $\neg (F_i^w(x) = y \in K^w)$
	\item[(B)]  V1.2 unless $b = 1$, $z = 11c(i,x,y)$ for $i,x,y\in\N$ with $i>0$ 
                and $\neg (F_i^w(x) = y \in \oli{K^w})$
    \item[(C)]  V3 unless $b = 0$ and $z = 10c(i,x,F_i^w(x))$ for $i>0$ and $x\in\N$ with $0<t_s(0,i)\le z$.
    \item[(D)]  V5 unless $b = 0$ and $z = 11c(i,x,F_i^w(x))$ for $i>0$ and $x\in\N$ with $0<t_s(1,i)\le z$.
\end{itemize}
This proves statement~\ref{item_389217401}. Let us argue for statement~\ref{item_4712800032}. According to
(A) and (B) $w0$ satisfies V1. If $w0$ does not satisfy V3 (resp., V5), then according to (C) (resp., (D)),
$z = 10c(i,x,F_i^w(x))$ (resp., $z = 11c(i,x,F_i^w(x))$) for $i\in\N^+$ and $x\in\N$ as well as $0< t_s(0,i)\le z$
(resp., $0 < t_s(1,i)\le z$). However, this case is covered by statement~\ref{item_01283470}
(resp., statement~\ref{item_10278480}). This proves statement~\ref{item_4712800032}.

Let us consider statements~\ref{item_01283470} and \ref{item_10278480} simultaneously. 
Due to the statements~(C) and (D) it suffices to argue for V1.1 (resp., V1.2 when 
arguing for statement~\ref{item_10278480}). 
Here it is sufficient to show $F_i^w(x)\in K^w$ (resp., $F_i^w(x)\in \oli{K^w}$).
For a contradiction assume $F_i^w(x)\notin K^w$ (resp., $F_i^w(x)\notin \oli{K^w}$). Let $s'>0$ be the step with
$T(s') = (0,i)$ (resp., $T(s') = (1,i)$). Then $s' \le s$. By Claim~\ref{claim_10794190247210},
the oracle $w$ is $t_{s'-1}$-valid and by Claim~\ref{claim_041730745333}, 
$F_i^w(x)$ is defined and $F_i^w(x)\notin K^v$ (resp., $F_i^w(x)\notin\oli{K^v}$)
for all $v\sqsupseteq w$. Hence, $w$ is even
$t$-valid for $t = t_{s'-1}\cup\{(0,i)\mapsto 0\}$ (resp., $t = t_{s'-1}\cup\{(1,i)\mapsto 0\}$).
But then the construction would have chosen $t_{s'} = t$, in contradiction to $t_s(0,i) > 0$
(resp., $t_s(1,i) > 0$).
\end{proof}

We now show that the described construction is possible: for a
contradiction, assume that it is not. Hence, there exists
a minimal $s>0$ such that step~$s$ fails. Then $w_{s-1}$
is $t_{s-1}$-valid. 

Assume that in step~$s$ some task $(a,i)$ for $a\in\{0,1\}$
and $i\in\N^+$ is treated. Then $t_{s-1}(a,i)$ is not defined
as this value is defined in the unique treatment of the task
$(a,i)$. Thus, $t'$ is well defined. Moreover, if there exists
a $t'$-valid oracle $v\sqsupsetneq w_{s-1}$, then step~$s$ is
clearly possible. Otherwise, by the (sufficiently large) choice
of $t_s(a,i)$, the oracle $w_{s-1}$ is even $t_s$-valid and by Claim~\ref{claim_82146969998}, 
there exists $b\in\{0,1\}$ such that the oracle
$w_{s} = w_{s-1}b$ is $t_s$-valid. Hence, if some task $(a,i)$ for $a\in\{0,1\}$ is treated 
in step~$s$, then we obtain a contradiction.

From now on we assume that step~$s$ treats some task $(2,i)$
for $i>0$. Thus, $t_s = t_{s-1}$ and we need to show that there exist some $t_s$-valid
$w_s\sqsupsetneq w_{s-1}$ and some $n\in\N$ such that the
computation $P_i^{w_s}(0^n)$ is defined, $0^n\in A^v\Leftrightarrow
0^n\in A^{w_s}$ for all $v\sqsupseteq w_s$, and $\big(0^n\in A^{w_s} \Leftrightarrow 
P_i^{w_s}(0^n)\text{ rejects}\big)$.

Choose $n$ large enough such that $2^n > 2(n^i + i)$. Let $u_0\sqsupseteq w_{s-1}$ be
the minimal $t_s$-valid oracle that is defined for all words of length $\le n$. 
Such an oracle exists by Claim~\ref{claim_82146969998}. Moreover, let $u\sqsupseteq u_0$
be the minimal $t_s$-valid oracle that is defined for all words of length 
$2(n^i + i)$. Such an oracle exists by Claim~\ref{claim_82146969998} and
by Claim~\ref{claim_82146969998}.\ref{item_389217401}, $u\cap 0\az^n = \emptyset$. If $P_i^u(0^n)$ accepts,
then it definitely accepts by the choice of $u$ and since $0^n\notin A^v$ for all $v\sqsupseteq u$
(note that $u$ is defined for all words of length $n+1$), we can
choose $w_s = u$ and obtain a contradiction to the assumption that step~$s$ is not
possible.

From now on we assume that $P_i^u(0^n)$ rejects. Let $U$ be the set of oracle queries
of $P_i^u(0^n)$ whose length is $\ge n + 1$. We define $Q_0(U) = U$ and for $m\in\N$
$$
Q_{m+1}(U) := \bigcup_{\substack{\{10c(j,x,y),11c(j,x,y)\}\cap Q_m(U)\ne\emptyset\\
j>0, x,y\in\N}} \{q\in\az^{\ge n+1}\mid \text{$q$ is queried by $F_j^u(x)$}\}.
$$
Moreover, define $Q(U) = \bigcup_{m\in\N} Q_m(U)$.
\begin{claim}\label{claim_18471850}
$\ell(Q(U)) \le 2(n^i + i)$.
\end{claim}
\begin{proof}
By definition of $Q_0(U)$ it holds $\ell(Q_0(U)) \le n^i + i$. We show that for all 
$m\in\N$ it holds $\ell(Q_{m+1}(U)) \le \nicefrac{\ell(Q_m(U))}{2}$. Then for all $m\in\N$ it holds
$\ell(Q_m(U)) \le \nicefrac{\ell(Q_0(U))}{2^m}$ and thus,
$$
\ell(\bigcup_{k = 0}^m Q_k(U)) \le \ell(Q_0(U)) \cdot \sum_{k = 0}^m\nicefrac{1}{2^k}
\le (n^i + i)\cdot \frac{1-\nicefrac{1}{2^{m+1}}}{\nicefrac{1}{2}} < 2(n^i + i),
$$
which shows $\ell(Q(U))\le 2(n^i + i)$.

It remains to show that $\ell(Q_{m+1}(U)) \le \nicefrac{\ell(Q_m(U))}{2}$ for all 
$m\in\N$. Let $\alpha\in Q_m(U)$. If $\alpha$ is not of the form $10c(j,x,y)$ or $11c(j,x,y)$,
then it generates no elements in $Q_{m+1}(U)$. Let
$\alpha = 1b c(j,x,y)$ for $b\in\{0,1\}$, $j\in\N^+$, and $x,y\in\N$. 
This affects that all queries of $F_j^u(x)$ are added into $Q_{m+1}(U)$.
The computation time of $F_j^u(x)$ (and also the sum of the lengths of all queries asked by that computation)
is bounded by $|x|^j + j\le \nicefrac{|c(j,x,y)|}{2}$ (cf.\ the definition
of $c(\cdot,\cdot,\cdot)$ and the definition of the pairing function). Hence,
\begin{eqnarray*}
\ell(Q_{m+1}(U))&= \ell\Big( \bigcup_{\substack{\{10c(j,x,y),11c(j,x,y)\}\cap Q_m(U)\ne\emptyset\\
j>0, x,y\in\N}} \{q\in\az^{\ge n+1}\mid \text{$q$ is queried by $F_j^u(x)$}\}\Big)\\
&\le \sum_{\substack{\{10c(j,x,y),11c(j,x,y)\}\cap Q_m(U)\ne\emptyset\\j>0, x,y\in\N}}
    \underbrace{\ell\big( \{q\in\az^{\ge n+1}\mid \text{$q$ is queried by $F_j^u(x)$}\}\big)}_{\le 
    \nicefrac{|c(j,x,y)|}{2}}\\
&\hspace{-21mm}\le \nicefrac{1}{2} \cdot \sum_{\substack{\{10c(j,x,y),11c(j,x,y)\}\cap Q_m(U)
    \ne\emptyset\\j>0, x,y\in\N}}|c(j,x,y)| \le \nicefrac{\ell(Q_m(U))}{2},
\end{eqnarray*}
which finishes the proof.
\end{proof}

As by the choice of $n$, it holds $|Q(U)| \le\ell(Q(U))\le 2(n^i+i) < 2^n$, 
there exists $\alpha\in 0\az^n$ that is not in $Q(U)$. 
Let $u'$ be the minimal $t_s$-valid oracle $\sqsupseteq u_0$ that is defined for all words $\le 01^n$ and 
satisfies $u'\cap 0\az^n = \{\alpha\}$.
Such an oracle exists by Claim~\ref{claim_82146969998}.\ref{item_389217401}. 

\begin{claim}\label{claim_13780781024}
There exists a $t_s$-valid oracle $v\sqsupsetneq u'$ that is defined for
all words of length $2(n^i + i)$ and satisfies $v(q) = u(q)$ for all
$q\in Q(U)$. 
\end{claim}
\begin{proof}
As $\alpha\notin Q(U)$ it holds
$u'(q) = u(q)$ for all $q\in Q(U)$ that $u'$ is defined for.

It suffices to show the following:
\begin{equation}\label{eq_1}
\parbox[c]{128mm}{
    For each $t_s$-valid $w\sqsupseteq u'$ with $w(q) = u(q)$ for all $q\in Q(U)$
    that $w$ is defined for, there exists $b\in\{0,1\}$ such that $wb$ is $t_s$-valid
    and $wb(q) = u(q)$ for all $q\in Q(U)$ that $wb$ is defined for.
}
\end{equation}
Let some $w$ with the properties of (\ref{eq_1}) be given. Moreover, let $z = |w|$,
i.e., $z$ is the least word that $w$ is not defined for.
We study three cases.
\begin{enumerate}[wide, labelwidth=!, labelindent=0pt]
    \item Assume $z = 1ac(j,x,F_j^w(x))$ for $a\in\{0,1\}$, $j\in\N^+$, and
    $x,y\in\N$ with $0< t_s(a,j) \le z$. Then choose $b = 1$.
    According to Claim~\ref{claim_82146969998}.\ref{item_01283470} (resp., 
    Claim~\ref{claim_82146969998}.\ref{item_10278480} in case $a = 1$) the oracle $wb$ is $t_s$-valid.
    
    It remains to show that $z\in Q(U)\Rightarrow z\in u$. For a contradiction assume
    $z\in Q(U) \wedge z\notin u$. Then $F_j^u(x) \ne F_j^w(x)$, since
    $F_j^u(x) = F_j^w(x)$ and $0<t_s(a,j) \le z$ would imply by V3 (resp., V5 in case $a = 1$), 
    that $1ac(j,x,y)\in u$. Hence, $F_j^w(x)\ne F_j^u(x)$, which shows that there is some query
    $q\in u\triangle w$ that is asked by both computations $F_j^w(x)$ and $F_j^u(x)$ (otherwise,
    the two queries would output the same value). 
    In particular, $q\in Q(U)$. As $|q| \le |x|^j + j < |c(j,x,y)| < c(j,x,y)$,
    the oracle $w$ is defined for $q$ and by assumption $w(q) = u(q)$, a contradiction.
    
    \item If $z = 0y$ for $y\in\az^m$ and $m\in\N$, then we choose $b = u(z)$ and hence,
    $wb(q) = u(q)$ for all $q\in Q(U)$ that $wb$ is defined for. Moreover, by
    Claim~\ref{claim_82146969998}.\ref{item_389217401}, $wb$ is $t_s$-valid.
    
    \item For the remaining cases choose $b = 0$. Then Claim~\ref{claim_82146969998}.\ref{item_4712800032}
    states that $wb$ is $t_s$-valid. It remains to show that $z\in Q(U) \Rightarrow z\notin u$.

    For a contradiction assume $z\in Q(U)\cap u$. Let $u''$ be the prefix of $u$ that is defined
    for exactly the words $< z$. As $w_{s-1}\sqsubseteq u''\sqsubseteq u$ and $w_{s-1}$ as well as
    $u$ are $t_s$-valid, $u''$ is $t_s$-valid as well by Claim~\ref{claim_1094101243}. 
    \begin{itemize}
        \item Assume that
    Claim~\ref{claim_82146969998}.\ref{item_01283470} or
    Claim~\ref{claim_82146969998}.\ref{item_10278480} can be applied to $u''$. Then
    $z = 1ac(j,x,F_j^{u''}(x))$ for $a\in\{0,1\}$, $j\in\N^+$, and $x\in\N$ with
    $0<t_s(a,i)\le z$. By Claim~\ref{claim_041730745333}, $F_j^u(x) = F_j^{u''}(x)$, which implies
    $F_j^u(x)\ne F_j^w(x)$ (otherwise, we were in a case that has already been treated). This
    shows that there is some query
    $q\in u\triangle w$ that is asked by both computations $F_j^w(x)$ and $F_j^u(x)$ (otherwise,
    the two computations would output the same value). 
    In particular, $q\in Q(U)$. As $|q| \le |x|^j + j < |c(j,x,y)| < c(j,x,y)$,
    the oracle $w$ is defined for $q$ and by assumption $w(q) = u(q)$, a contradiction. 
        \item Now assume that Claim~\ref{claim_82146969998}.\ref{item_389217401} or 
    Claim~\ref{claim_82146969998}.\ref{item_4712800032} can be applied to $u''$ and yields that
    $u''0$ is $t_s$-valid. By Claim~\ref{claim_82146969998}, $u''0$ can be extended to a $t_s$-valid
    oracle $v'$ defined for exactly the words of length $\le 2(n^i + i)$. As $u$ and $v'$ agree on all
    words $< z$ and $v'(z) = 0 < 1 = u(z)$, it holds $v'< u$, in contradiction to the choice of $u$.
    \end{itemize}
    In both cases we obtain a contradiction. Hence, $u(q) = w0(q)$ for all $q\in Q(U)$ that $w1$ is defined for.
    
\end{enumerate}

In all cases (\ref{eq_1}) holds. This completes the proof of Claim~\ref{claim_13780781024}.
\end{proof}

Recall that $P_i^u(0^n)$ rejects. Let $v$ be the oracle postulated by Claim~\ref{claim_13780781024}. 
Since all queries of $P_i^u(0^n)$ are in $U\subseteq Q(U)$ and by Claim~\ref{claim_13780781024},
$u$ and $v$ agree on all these queries, the computation $P_i^v(0^n)$ rejects as well. Moreover,
this computation is defined as $v$ is defined for all words of length $2(n^i + i)$. However, as $\alpha\in v$,
we obtain $0^n\in A^{v'}$ for all $v'\sqsupseteq v$, 
which is a contradiction to the assumption that the construction fails in step~$s$
treating the task $(2,i)$.

We now have seen that the construction described above is possible. It remains to prove that
\begin{itemize}
	\item $\cNP^O\ne\cP^O$,
    \item $K^O$ has $\cP^O$-optimal proof systems, and
    \item $\oli{K^O}$ has $\cP^O$-optimal proof systems.
\end{itemize}
This is shown in the next three claims.

\begin{claim}
$\cNP^O\ne\cP^O$.
\end{claim}
\begin{proof}
Assume $\cNP^O = \cP^O$. Then there exists $i>0$ such that $L(P_i^O) = A^O$.
Let $s$ be the step with $T(s) = (2,i)$. By construction, there exists some $n\in\N$ such that the
        computation $P_i^{w_s}(0^n)$ is defined, $0^n\in A^v\Leftrightarrow
        0^n\in A^{w_s}$ for all $v\sqsupseteq w_s$, and $\big(0^n\in A^{w_s} \Leftrightarrow 
        P_i^{w_s}(0^n)\text{ rejects}\big)$.
Hence, $0^n\in A^O$ if and only if $P_i^O(0^n)$ definitely rejects.
This contradicts $L(P_i^O) = A^O$.
\end{proof}

\begin{claim}
$K^O$ has $\cP^O$-optimal proof systems.
\end{claim}
\begin{proof}
Let $g$ be a proof system for $K^O$ and $a\in K^O$.
Define
$$
f(z) = \begin{cases}
y&\text{if $z = 010c(i,x,y)$ and $10c(i,x,y)\in O$ for $i\in\N^+$ and $x,y\in\N$}\\
g(y)&\text{if $z = 1y$}\\
a&\text{otherwise}
\end{cases}
$$
Then $f\in\FP^O$ and $f(\N)\supseteq K^O$ as $g$ is a proof system for $K^O$.
We show $f(\N) \subseteq K^O$. As $g$ is a proof system for $K^O$ and $a\in K^O$, 
it suffices to show $f(z)\in K^O$ for $z = 010c(i,x,y)$ with $10c(i,x,y)\in O$, $i\in\N^+$, and $x,y\in\N$.
Let $s$ be large enough such that $w_s$ is defined for $10c(i,x,y)$. Then by V1 and Claim~\ref{claim_041730745333}, 
$y\in K^v$ for all
$v\sqsupseteq w_s$. It follows $f(z) = y\in K^O$ and thus, $f$ is a proof system for $K^O$.

In order to show that $f$ is $\cP^O$-optimal, let $h$ be an arbitrary proof system for $K^O$.
Then there exists $i\in\N^+$ such that $F_i^O$ computes $h$. Let $s$ be the step with $T(0,i) = s$.
It holds $t_s(0,i) > 0$ (otherwise, by V2 there exists $x$ such that $F_i^w(x)$ is defined  
                and $F_i^w(x)\notin K^v$ for all $v\sqsupseteq w$, which would imply that $F_i^O$
                is not a proof system for $K^O$).
Define
$$
\pi(x) = \begin{cases}
010c(i,x,F_i^O(x))&\text{if $10c(i,x,F_i^O(x)) \ge t_s(0,i)$}\\
z&\text{if $10c(i,x,F_i^O(x)) < t_s(0,i)$ and $z$ is minimal with $f(z) = F_i^O(x)$}
\end{cases}
$$
$\pi$ is total as $f$ and $F_i^O$ are proof systems for $K^O$ and thus, for each
$x$ there exists $z$ with $f(z) = F_i^O(x)$.
Moreover, since $t_s(0,i)$ is a constant, $\pi\in \FP^O$. It remains to show $F_i^O(x) = f(\pi(x))$
for all $x$. For all $x$ with $10c(i,x,F_i^O(x)) < t_s(0,i)$, this clearly holds.
Assume $10c(i,x,F_i^O(x)) \ge t_s(0,i)$. Choose $s'\ge s$ large enough such that
$w_{s'}$ is defined for $10c(i,x,F_i^O(x))$. By Claim~\ref{claim_041730745333},
the computation $F_i^{w_{s'}}(x)$ is defined and hence, $F_i^{w_{s'}}(x) = F_i^O(x)$. 
Then by V3, $10c(i,x,F_i^O(x)) = 10c(i,x,F_i^{w_{s'}}(x)) \in w_{s'}\subseteq O$ and
thus, $f(\pi(x)) = f(010c(i,x,F_i^O(x))) = F_i^O(x)$, which shows that $F_i^O$ is $\cP^O$-simulated
by $f$. This completes the proof.
\end{proof}

\begin{claim}
$\oli{K^O}$ has $\cP^O$-optimal proof systems.
\end{claim}
\begin{proof}
Let $g$ be a proof system for $\oli{K^O}$ and $a\in \oli{K^O}$.
Define
$$
f(z) = \begin{cases}
y&\text{if $z = 011c(i,x,y)$ and $11c(i,x,y)\in O$ for $i\in\N^+$ and $x,y\in\N$}\\
g(y)&\text{if $z = 1y$}\\
a&\text{otherwise}
\end{cases}
$$
Then $f\in\FP^O$ and $f(\N)\supseteq \oli{K^O}$ as $g$ is a proof system for $\oli{K^O}$.
We show $f(\N) \subseteq \oli{K^O}$. As $g$ is a proof system for $\oli{K^O}$ and $a\in \oli{K^O}$, 
it suffices to show $f(z)\in \oli{K^O}$ for $z = 011c(i,x,y)$ with $11c(i,x,y)\in O$, $i\in\N^+$, and $x,y\in\N$.
Let $s$ be large enough such that $w_s$ is defined for $11c(i,x,y)$. Then by V1 and Claim~\ref{claim_041730745333}, 
$y\in \oli{K^v}$ for all
$v\sqsupseteq w_s$. It follows $f(z) = y\in \oli{K^O}$ and thus, $f$ is a proof system for $\oli{K^O}$.

In order to show that $f$ is $\cP^O$-optimal, let $h$ be an arbitrary proof system for $\oli{K^O}$.
Then there exists $i\in\N^+$ such that $F_i^O$ computes $h$. Let $s$ be the step with $T(1,i) = s$.
It holds $t_s(1,i) > 0$ (otherwise, by V2 there exists $x$ such that $F_i^w(x)$ is defined  
                and $F_i^w(x)\in \oli{K^v}$ for all $v\sqsupseteq w$, which would imply that $F_i^O$
                is not a proof system for $\oli{K^O}$).
Define
$$
\pi(x) = \begin{cases}
011c(i,x,F_i^O(x))&\text{if $11c(i,x,F_i^O(x)) \ge t_s(1,i)$}\\
z&\text{if $11c(i,x,F_i^O(x)) < t_s(1,i)$ and $z$ is minimal with $f(z) = F_i^O(x)$}
\end{cases}
$$
$\pi$ is total as $f$ and $F_i^O$ are proof systems for $\oli{K^O}$ and thus, for each
$x$ there exists $z$ with $f(z) = F_i^O(x)$.
Moreover, since $t_s(1,i)$ is a constant, $\pi\in \FP^O$. It remains to show $F_i^O(x) = f(\pi(x))$
for all $x$. For all $x$ with $11c(i,x,F_i^O(x)) < t_s(1,i)$, this clearly holds.
Assume $11c(i,x,F_i^O(x)) \ge t_s(1,i)$. Choose $s'\ge s$ large enough such that
$w_{s'}$ is defined for $11c(i,x,F_i^O(x))$. By Claim~\ref{claim_041730745333},
the computation $F_i^{w_{s'}}(x)$ is defined and hence, $F_i^{w_{s'}}(x) = F_i^O(x)$. 
Then by V5, $11c(i,x,F_i^O(x)) = 11c(i,x,F_i^{w_{s'}}(x)) \in w_{s'}\subseteq O$ and
thus, $f(\pi(x)) = f(011c(i,x,F_i^O(x))) = F_i^O(x)$, which shows that $F_i^O$ is $\cP^O$-simulated
by $f$. This completes the proof.
\end{proof}

 This completes the proof of Theorem~\ref{theorem_0917240914}.
\end{proof}

\bibliographystyle{alpha}


\end{document}